\documentclass[11pt]{article}
\usepackage{graphicx,amssymb,amsmath,amsthm}
\usepackage[linesnumbered, ruled]{algorithm2e}
\SetKwInput{KwResult}{Function}

\usepackage{amsfonts}
\usepackage{lineno}
\usepackage{booktabs}
\usepackage{cite}
\usepackage{color}
\usepackage{enumerate}
\usepackage{float}
\usepackage{mathrsfs}
\usepackage{subfig}
\usepackage{tabularx}
\usepackage{fullpage}
\usepackage[top=0.8in, bottom=0.9in, left=0.9in, right=0.9in]{geometry}
\usepackage[pdftex, plainpages=false, pdfpagelabels, 
             bookmarks=false,
             bookmarksopen=true,
             bookmarksnumbered=true,
             breaklinks=true,
             linktocpage,
             pagebackref,
             colorlinks=true,
             linkcolor=blue,
             urlcolor=cyan,
             citecolor=red,
             anchorcolor=green,
             hyperindex=true,
             hyperfigures]{hyperref}

\def\calC{\mathcal{C}}
\def\calU{\mathcal{U}}

\def\calA{\mathcal{A}}
\def\hatA{\hat{A}}
\def\hatcalA{\hat{\mathcal{A}}}
\def\calL{\mathcal{L}}

\def\calI{\mathcal{I}}

\def\opta{A_{\text{opt}}}
\def\optp{P_{\text{opt}}}

\newtheorem{lemma}{Lemma}
\newtheorem{theorem}{Theorem}

\newtheorem{corollary}{Corollary}

\title{An Optimal Algorithm for Half-plane Hitting Set\thanks{This research was supported in part by NSF under Grant CCF-2300356. This paper will appear in {\em Proceedings of the SIAM Symposium on Simplicity in Algorithms (SOSA 2025)}.}
}

\author{
Gang Liu\thanks{Kahlert School of Computing,
University of Utah, Salt Lake City, UT 84112, USA. {\tt u0866264@utah.edu}}
\and
Haitao Wang\thanks{Kahlert School of Computing,
University of Utah, Salt Lake City, UT 84112, USA. {\tt haitao.wang@utah.edu}}
}

\begin{document}
\pagestyle{plain}
\date{}

\thispagestyle{empty}
\maketitle

\vspace{-0.3in}

\begin{abstract}
Given a set $ P $ of $n$ points and a set $ H $ of $n$ half-planes in the plane, we consider the problem of computing a smallest subset of points such that each half-plane contains at least one point of the subset. The previously best algorithm solves the problem in $O(n^3 \log n)$ time. It is also known that $\Omega(n \log n)$ is a lower bound for the problem under the algebraic decision tree model.
In this paper, we present an $O(n \log n)$ time algorithm, which matches the lower bound and thus is optimal. Another virtue of the algorithm is that it is relatively simple.
\end{abstract}


\section{Introduction}
\label{sec:intro}
Given a set $P$ of $n$ points and a set $H$ of $n$ half-planes in the plane, we say that a point {\em hits} a half-plane if the half-plane contains the point. A {\em hitting set} for $H$ is a subset $P'$ of $P$ such that each half-plane of $H$ is hit by at least one point of $P'$. Given $P$ and $H$, the {\em half-plane hitting set problem} is to compute a smallest hitting set for $H$. 

The problem was studied before~\cite{ref:ChanEx14, ref:Har-PeledWe12, ref:LiuGe23, ref:LiuOn24, ref:LiuUn24}. Har-Peled and Lee~\cite{ref:Har-PeledWe12} first proposed an $O(n^6)$ time algorithm. Liu and Wang~\cite{ref:LiuGe23} gave a faster algorithm by reducing the problem to $O(n^2)$ instances of the {\em lower-only} half-plane hitting set problem in which all half-planes are lower ones. Consequently, if the lower-only problem can be solved in $O(T)$ time, the general problem is solvable in $O(n^2 \cdot T)$ time. Liu and Wang~\cite{ref:LiuGe23} derived an algorithm of $O(n^2 \log n)$ time for the lower-only problem and thus also solved the general half-plane hitting set problem in $O(n^4\log n)$ time. 
Very recently Wang and Xue~\cite{ref:WangAl24} gave an improved $O(n \log n)$ time algorithm for the lower-only problem, leading to a solution for the general problem in $O(n^3 \log n)$. 
On the other hand, Wang and Xue~\cite{ref:WangAl24} proved a lower bound of $\Omega(n \log n)$ under the algebraic decision tree model even for the lower-only problem.

In this paper, we present a new algorithm for the general half-plane hitting set problem and our algorithm runs in $O(n \log n)$ time. Our algorithm not only significantly improves the previous $O(n^3\log n)$ time result by a quadratic factor, but also matches the $\Omega(n\log n)$ lower bound and thus is optimal. In addition, our algorithm is also interesting because it is relatively simple.

\paragraph{\bf Related work.}
In the {\em weighted} half-plane hitting set problem, every point of $P$ has a weight and the goal is to compute a hitting set for $H$ that has the smallest total weight. The problem was also studied before. The above $O(n^6)$ time algorithm in~\cite{ref:Har-PeledWe12} also works for the weighted case. So does the reduction by Liu and Wang~\cite{ref:LiuGe23} discussed above. Consequently, if the weighted lower-only problem has an $O(T)$ time solution, then the general weighted problem can be solved in $O(n^2 \cdot T)$ time. In another recent paper Liu and Wang~\cite{ref:LiuOn24} presented an $O(n^{3/2}\log^2 n)$ time algorithm for the weighted lower-only problem and therefore solved the weighted general half-plane hitting set problem in $O(n^{7/2}\log^2 n)$ time. 

A closely related problem is the half-plane {\em coverage} problem. Given $P$ and $H$ as above, the problem is to find a smallest subset of $H$ whose union covers all the points of $P$. Notice that the {\em lower-only} problem, where all half-planes are lower ones, is dual to the lower-only hitting set problem (i.e., the two problems can be reduced to each other in linear time). As such, an algorithm for the lower-only hitting set problem can be used to solve the lower-only coverage problem with the same time complexity, and vice versa. In fact, the $O(n\log n)$ time algorithm of Wang and Xue~\cite{ref:WangAl24} discussed above for the lower-only hitting set problem was originally described to solve the lower-only coverage problem. This is also the case for the weighted problem and therefore the $O(n^{3/2}\log^2 n)$ time algorithm of Liu and Wang~\cite{ref:LiuOn24} for the weighted lower-only half-plane hitting set problem also works for the weighted lower-only half-plane coverage problem. It should be noted that in the general case where both upper and lower half-planes are present in $H$, the hitting set problem and the coverage problem are not dual to each other anymore. For the general half-plane coverage problem, the currently fastest algorithm was due to Wang and Xue~\cite{ref:WangAl24}  and their algorithm runs in $O(n^{4/3} \log^{5/3} n \log^{O(1)} \log n)$ time. 
This actually makes our $O(n\log n)$ time algorithm for the general half-plane hitting set problem all the more interesting. 

As a fundamental problem, the hitting set problem has been studied extensively in the literature. Hitting set problems under various geometric settings have also attracted much attention and many of these problems are NP-hard~\cite{ref:OualiA14,ref:ChanEx14,ref:MustafaIm10, ref:BusPr18, ref:EvenHi05, ref:GanjugunteGe11, ref:LiA15}. For example, given a set of disks and a set of points in the plane, the disk hitting set problem is to find a smallest subset of points that hit all disks. The problem is NP-hard, even when all disks have the same radius \cite{ref:DurocherDu15, ref:KarpRe72, ref:MustafaIm10}. 
Certain variants of the problem have polynomial time solutions. Liu and Wang \cite{ref:LiuGe23} considered a line-constrained version of the problem in which the centers of all disks lie on a line, and presented a polynomial time algorithm. See \cite{ref:LiuOn24,ref:LiuUn24,ref:PedersenOn18,ref:PedersenAl22} also for the related coverage problem under similar settings. 


\paragraph{\bf Our approach.}
To solve the half-plane hitting set problem on $P$ and $H$, instead of reducing the problem to $O(n^2)$ instances of the lower-only case as in the previous work~\cite{ref:LiuGe23}, we propose a new method that 
reduces the problem to the following {\em circular-point coverage} problem: Given a set $\calA$ of arcs and a set $B$ of points on a circle $C$, compute a smallest subset of arcs whose union covers all points. The circular-point coverage problem can be solved in $O((|\calA| + |B|) \log (|\calA| + |B|))$ time~\cite{ref:WangAl24}, by a slight modification of the algorithms for a circle coverage problem~\cite{ref:AgarwalCo24,ref:LeeOn84} (whose goal is to cover the entire circle). In our reduction, each half-plane of $H$ defines a single point of $B$ and  $|B|=n$. Each point of $P$, however, may generate as many as $n/2$ arcs of $\calA$; thus we have $|\calA|=O(n^2)$. Our reduction makes sure that a point $p \in P$ hits a half-plane $h \in H$ if and only if the point of $B$ defined by $h$ is covered by one of the arcs generated by $p$. 
An optimal solution to the hitting set problem on $P$ and $H$ can be easily obtained from an optimal solution to the circular-point coverage problem: 
Suppose that $\opta$ is a smallest subset of arcs of $\calA$ whose union covers $B$. For each arc of $\opta$, if the arc is generated by the point $p\in P$, then we add $p$ to $\optp$. We prove that the subset $\optp$ thus obtained is a smallest hitting set of $H$. 

The above solves the half-plane hitting set problem in $O(n^2\log n)$ time since $|\calA|=O(n^2)$. We  further improve the algorithm  
by showing that it suffices to use a small subset $\hat{\calA}\subseteq \calA$ of size at most $4n$ and we prove that
a smallest subset of $\hatcalA$ for covering $B$ is also a smallest subset of $\calA$ for covering $B$. As such, with $\hatcalA$ and $B$, the circular-point coverage problem can now be solved in $O(n\log n)$ time since $|\hatcalA|\leq 4n$. In addition, we develop an algorithm to compute $\hatcalA$ in $O(n\log n)$ time. All these efforts lead to an $O(n\log n)$ time algorithm for the half-plane hitting set problem. 

\paragraph{\bf Outline.} The rest of the paper is organized as follows. After defining notation in Section~\ref{sec:pre}, we introduce our problem reduction in Section~\ref{sec:trans}. The algorithm is then described in Section~\ref{sec:algo2}. 

\section{Preliminaries}
\label{sec:pre}
Let $P$ be a set of $n$ points in the plane and $H$ a set of $n$ half-planes. We assume that each half-plane must be hit by a point of $P$; otherwise a hitting set does not exist. We can check whether this is true in $O(n\log n)$ time (e.g., first compute the convex hull of $P$; then for each half-plane $h$ we can use the convex hull to determine whether $h$ contains a point of $P$ in $O(\log n)$ time). 

For each half-plane $h$, we define its {\em normal} as a vector perpendicular to its bounding line and toward the interior of $h$. If two half-planes $h,h'\in H$ have the same normal, then one of the half-plane contains the other, say, $h\subseteq h'$. As such, $h'$ is redundant because any point hitting $h$ must $h'$ as well. We can easily find all redundant half-planes in $O(n\log n)$ time by sorting all half-planes by their normals. In what follows, we assume that no two half-planes have the same normal. 

To make our later discussion easier, we first determine whether $P$ has a point $p$ that hits all half-planes. If such a point $p$ exists, then $p$ itself forms a smallest hitting set for $H$ and $p$ must be in the common intersection $D$ of all half-planes. This special case can be solved in $O(n\log n)$ time as follows. First, we compute $D$, which can be done in $O(n\log n)$ time, e.g., by an incremental algorithm. If $D=\emptyset$, then there is no point in $P$ that hits all half-planes. If $D\neq \emptyset$, then for each point $p\in P$, we determine whether $p$ is inside $D$, which takes $O(\log n)$ time as $D$ is a convex polygon. As such, in $O(n\log n)$ time, we can solve the special case in which $P$ has a point that hits all half-planes. In the following, we assume that no point of $P$ hits all half-planes. 


\section{Reducing to a circular-point coverage problem}
\label{sec:trans}

In this section, we reduce the half-plane hitting set problem for $P$ and $H$ to an instance of the circular-point coverage problem for a set $\calA$ of arcs and a set $B$ of points on a circle $C$. 

In the following, we first define the circular-point coverage problem, i.e., $\calA$, $B$, and $C$. Then we prove the correctness of the reduction, i.e., explain why a solution to the circular-point coverage problem leads to a solution to the half-plane hitting set problem. 

\subsection{Defining the circular-point coverage problem}

First, let $C$ denote an arbitrary unit circle.

\paragraph{Defining $B$: the points.}
For each half-plane $h\in H$, we define a point $b$ as the point of $C$ intersected by the ray originating from the center of $C$ and parallel to the normal of $h$. We say that $b$ is {\em defined} by $h$ and $h$ is the {\em defining} half-plane of $b$. 
Let $B$ denote the set of all points on $C$ defined by the half-planes of $H$. As no two half-planes of $H$ have the same normal, no two points of $B$ share the same location. 

We sort the points of $B$ counterclockwise on $C$ as a cyclic list $b_1,b_2,\ldots,b_n$.  We use $B[i,j]$ to denote the counterclockwise (contiguous) sublist of the cyclic list of $B$ from $b_i$ to $b_j$ including both $b_i$ and $b_j$. In other words, if $i\leq j$, then $B[i,j]=\{b_i,b_{i+1},\ldots,b_j\}$; otherwise $B[i,j]=\{b_i,\ldots,b_n,b_1,\ldots,b_j\}$. 
For any two points $b,b'\in C$, we use $C[b,b']$ to denote the portion of $C$ counterclockwise from $b$ to $b'$ and including both $b$ and $b'$.


\paragraph{\bf Defining $\calA$: the arcs.} For each point $p\in P$, we define a set $A(p)$ of arcs on $C$. 
For each maximal sublist $B[i,j]$ of $B$ such that $p$ hits all defining half-planes of the points of $B[i,j]$, we add the arc $C[b_i,b_j]$ to $A(p)$. Due to the assumption that no point of $P$ hits all the half-planes, the set $A(p)$ is well defined. 
Note that $A(p)$ may have at most $\lfloor n/2\rfloor$ arcs and all these arcs are pairwise disjoint. Define $\calA$ to be the union of $A(p)$ of all point $p \in P$. Clearly, $|\calA|=O(n^2)$. 

\subsection{Correctness of the reduction}

Consider the circular-point coverage problem for $B$ and $\calA$. Suppose $\opta$ is an optimal solution, i.e., $\opta$ is a smallest subset of $\calA$ covering $B$. We create a subset $\optp$ of $P$ as follows: For each arc in $\opta$, we add its defining point to $\optp$. 

In what follows, we prove that $\optp$ is an optimal solution to our half-plane hitting set problem for $P$ and $H$. At first glance, one potential issue is that two arcs of $\opta$ might be defined by the same point of $P$ and thus a point may be added to $\optp$ multiple times. We will show that this is not possible, implying $|\optp|=|\opta|$. All these are proved in Corollary~\ref{corollary:circular-point}, which follows mostly from the following lemma.

\begin{lemma}
\label{lemm:circular-point}
    $B$ can be covered by $k$ arcs of $\calA$ if and only if $H$ can be hit by $k$ points of $P$.
\end{lemma}
\begin{proof}
Suppose $B$ can be covered by a subset $A'\subseteq \calA$ of $k$ arcs. Then, let $P'$ be the set of defining points of all arcs of $A'$. Clearly, $|P'|\leq k$. We claim that $P'$ is a hitting set of $H$. To see this, consider a half-plane $h$. Let $b$ be the point of $B$ defined by $h$. Then, $b$ is covered by an arc $\alpha$ of $A'$. By definition, $h$ is hit by the defining point of $\alpha$, which is in $P'$. Therefore, $P'$ is a hitting set of $H$. This proves one direction of the lemma. 
In the following, we prove the other direction. 

Suppose $H$ has a hitting set $P'$ of $k$ points. Our goal is to show that $\calA$ has $k$ arcs  whose union covers $B$. Since by our assumption no point of $P$ hits all the half-planes, $k\geq 2$ holds. 
Depending on whether $k=2$ or $k\geq 3$, there are two cases. 

\begin{enumerate}
    \item If $k=2$, let $p_u$ and $p_l$ denote the two points of $P'$. We rotate the coordinate system so that the line segment $\overline{p_up_l}$ is vertical and $p_u$ is higher than $p_l$. Since every half-plane of $H$ is hit by $p_u$ or $p_l$, it is not difficult to see that all upper half-planes (in the rotated coordinate system) must be hit by $p_u$ and all lower half-planes must be hit by $p_l$. Observe also that the points of $B$ defined by all upper (resp., lower) half-planes of $H$ form a sublist $B_u$ (resp., $B_l$) of $B$. Hence, $p_u$ must define an arc $A_u$ that covers all points of $B_u$ and $p_l$ must define an arc $A_l$ that covers all points of $B_l$. Hence, $A_u$ and $A_l$ together cover all points of $B$. Therefore, $\calA$ has two arcs whose union covers $B$. 

    \item If $k\geq 3$, then let $p_1,p_2,\ldots,p_t$ be the vertices of the convex hull of $P'$ ordered counterclockwise on the convex hull. Hence, $t\leq k$. For each pair of adjacent points $p_i$ and $p_{i+1}$ on the convex hull, with index module $t$, define $h_{i,i+1}$ as the half-plane whose bounding line contains the line segment $\overline{p_ip_{i+1}}$ such that it does not contain the interior of the convex hull of $P'$ (see Fig.~\ref{fig:convexhull}). We define a point $b'_{i,i+1}$ on the circle $C$ using the normal of $h_{i,i+1}$. The $t$ points $b'_{i,i+1}$, $1\leq i\leq t$, partition $C$ into $t$ arcs $A'_i$, where $A'_i=C[b'_{i-1,i},b'_{i,i+1}]$. In what follows, we show that for each arc $A'_i$, $\calA$ has an arc that covers all points of $B\cap A_i'$. This will prove that $\calA$ has $k$ arcs whose union covers $B$ since $t\leq k$. 

    Consider an arc $A_i'=C[b'_{i-1,i},b'_{i,i+1}]$. For any point $b\in A_i'$, let $h$ be its defining half-plane. Let $\rho$ denote the normal of $h$. By definition, $\rho$ is in the interval of the directions counterclockwise from the normal of $h_{i-1,i}$ to that of $h_{i,i+1}$ (see Fig.~\ref{fig:convexhull}). It is not difficult to see that the point $p_i$ is a most extreme point of $P'$ along the direction $\rho$. Since $h$ is hit by a point of $P'$, $h$ must be hit by $p_i$. This means that $b$ must be covered by an arc defined by $p_i$. As all points of $A_i'\cap B$ are contiguous in the cyclic list of $B$, they must be covered by a single arc defined by $p_i$. This proves that $\calA$ has an arc that covers all points of $B\cap A_i'$. 
\end{enumerate}

\begin{figure}[h]
\centering
\includegraphics[height=1.7in]{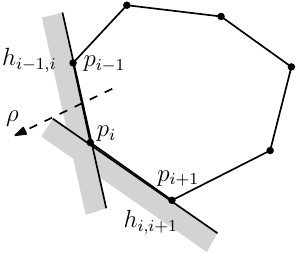}
\caption{Illustration of $h_{i-1,i}$, $h_{i,i+1}$, and $\rho$.}
\label{fig:convexhull}
\end{figure}

The lemma thus follows. 
\end{proof}


\begin{corollary}
\label{corollary:circular-point}
\begin{enumerate}
    \item The size of a smallest subset of $\calA$ for covering $B$ is equal to the size of a smallest subset of $P$ for hitting $H$.
    \item No two arcs of $\opta$ are defined by the same point.
    \item $\optp$ is an optimal solution to the half-plane hitting set problem.
\end{enumerate}
\end{corollary}
\begin{proof}
The first corollary statement directly follows from Lemma~\ref{lemm:circular-point}.

For the second corollary statement, notice that $\optp$ is a hitting set of $H$, which follows a similar argument to the first direction of Lemma~\ref{lemm:circular-point}.
Assume to the contrary that $\opta$ has two arcs $\alpha$ and $\alpha'$ defined by the same point $p \in P$. Then, by the definition of $\optp$, $|\optp| < |\opta|$. As $\opta$ is a smallest subset of $\calA$ covering $B$ and $\optp$ is a hitting set of $H$, we obtain that the size of a smallest hitting set of $H$ is smaller than the size of a smallest subset of $\calA$ for covering $B$, a contradiction to the first corollary statement.

For the third corollary statement, since $|\optp| \leq |\opta|$, $\opta$ is a smallest subset of $\calA$ covering $\calC$, and $\optp$ is a hitting set of $H$, by the first corollary statement, $|\optp| = |\opta|$ and $\optp$ must be a smallest hitting set of $H$.
\end{proof} 

By Corollary~\ref{corollary:circular-point}, we have successfully reduced our half-plane hitting set problem to the circular-point coverage problem on $B$ and $\calA$. In Section~\ref{sec:algo2}, we present our algorithm based on this reduction. 

\section{Algorithm}
\label{sec:algo2}

We first present a straightforward algorithm and then introduce the improved algorithm for solving the half-plane hitting set problem, 
based on our problem reduction in Section~\ref{sec:trans}. We follow the notation in Section~\ref{sec:trans}, e.g., $\calA$, $B$, $C$, etc.

With Corollary~\ref{corollary:circular-point}, to solve our half-plane hitting set problem for $P$ and $H$, we can do the following. First, construct $B$ and $\calA$. Second, find a smallest subset $\opta\subseteq \calA$ of arcs to cover $B$. Third, obtain $\optp$ from $\opta$. The first step can be easily done in $O(n^2)$ time. The second step takes $O((|B| + |\calA|) \log(|B| + |\calA|))$ time~\cite{ref:WangAl24}, which is $O(n^2\log n)$ since $|B|=n$ and $|\calA|=O(n^2)$. The third step takes $O(n)$ time. The total time of the algorithm is thus $O(n^2\log n)$.

In the rest of this section, we present an improved algorithm that runs in $O(n\log n)$ time. The main idea is that when solving the circular-point coverage problem, it suffices to use at most four arcs from $A(p)$ for each point $p \in P$. More specifically, for each $p\in P$, we will define a subset $\hatA(p)\subseteq A(p)$ of at most four arcs. Let $\hatcalA=\bigcup_{p\in P}\hatA(p)$, which contains at most $4n$ arcs. We will show that a smallest subset of $\hatcalA$ for covering $B$ is also a smallest subset of $\calA$ for covering $B$. As such, with $\hatcalA$ and $B$, the circular-point coverage problem can now be solved in $O(n\log n)$ time since $|\hatcalA|\leq 4n$. We will also give an algorithm that can compute $\hatcalA$ in $O(n\log n)$ time. Therefore, the total time of the entire algorithm is $O(n\log n)$.   

In the following, we first define $\hatcalA$, then discuss the correctness, i.e., why $\hatcalA$ is sufficient for solving the circular-point coverage problem, and finally present the algorithm to compute $\hatcalA$.

\subsection{Defining $\hat{\calA}$}

For ease of discussion, we assume that no half-plane of $H$ has a vertical bounding line. With the assumption, every half-plane is either an upper one or a lower one. Denote by $H_u$ (resp., $H_l$) the set of all upper (resp., lower) half-planes of $H$. 
Without loss of generality, we assume that $b_1,b_2,\ldots,b_{|H_l|}$ are the points of $B$ defined by the half-planes of $H_l$ while other points of $B$ are defined by those of $H_u$. Let $t=|H_l|$ and $B_l=\{b_1,b_2,\ldots,b_{t}\}$ and $B_u=B\setminus B_l$. 
It is easy to see that all points of $B_l$ lie on the lower half circle $C_l$ of $C$ while the points of $B_u$ lie on the upper half circle $C_u$ of $C$. 

For each point $p\in P$, if $A(p)$ has an arc covering at least one of $b_1$ and $b_n$, then we add it to $\hatA(p)$ and we call it the {\em left arc}, denoted by $\alpha_1(p)$; if $A(p)$ has an arc covering at least one of $b_t$ and $b_{t+1}$, then we also add it to $\hatA(p)$ and we call it the {\em right arc}, denoted by $\alpha_2(p)$. 

The above defines at most two arcs for $\hatA(p)$. In the following, we define two other arcs of $\hatA(p)$ for $p$; one of them lies on $C_l$, which only covers points of $B_l$ and is called the {\em lower arc}, denoted by $\alpha_l(p)$, and the other lies on $C_u$, which only covers points of $B_u$ and is called the {\em upper arc}, denoted by $\alpha_u(p)$. See Fig.~\ref{fig:arc_types}.

Let $A_l(p)$ be the subset of arcs of $A(p)\setminus\{\alpha_1(p),\alpha_2(p)\}$ on $C_l$ and $A_u(p)$ the subset of arcs of $A(p)\setminus\{\alpha_l(p),\alpha_r(p)\}$ on $C_u$. By definition, $\{\alpha_1(p),\alpha_2(p)\}$, $A_l(p)$, and $A_u(p)$ form a partition of $A(p)$. The lower arc is in $A_l(p)$ and the upper arc is in $A_u(p)$. 

Definitions of the lower and upper arcs are symmetric, and therefore we only explain the lower arcs. As the lower arc only covers points of $B_l$, it suffices to only consider the lower half-plane set $H_l$. The way we define the lower arcs is inspired by the technique of Wang and Xue~\cite{ref:WangAl24} for solving the lower-only half-plane coverage problem.

\begin{figure}[t]
\centering
\includegraphics[height=2in]{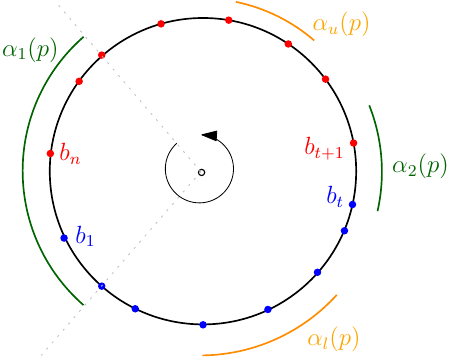}
\caption{Illustration of the relative positions of the four arcs of $\hatA(p)$. The red points are all on the upper half circle $C_u$ while the blue points are all on the lower half circle $C_l$.}
\label{fig:arc_types}
\end{figure}

\paragraph{Defining the lower arcs.}
Note that a point $p$ hits a lower half-plane $h\in H_l$ if and only if $p$ is below the bounding line of $h$.
We utilize the following commonly used duality in computational geometry~\cite{ref:deBergCo08}: A point $(a,b)$ is dual to a line $y=ax-b$, and a line $y=cx+d$ is dual to a point $(c,-d)$. For each point $p\in P$, let $p^*$ denote the lower half-plane whose bounding line is dual to $p$. Define $P^*=\{p^*\ |\ p\in P\}$. For each half-plane $h_i\in H_l$, $1\leq i\leq t$, let $h_i^*$ denote the point dual to the bounding line of $h_i$. Define $H_l^*=\{h_i^*\ |\ h_i\in H_l\}$. According to the duality, a point $p\in P$ hits a half-plane $h_i\in H_l$ if and only if $p^*$ contains $h_i^*$.

By duality, the points $h_1^*,h_2^*,\ldots,h_t^*$ of $H_l^*$ are sorted from left to right, which is consistent with the index order of the points $b_1,b_2,\ldots,b_t$ on $C_l$.  Depending on the context, $H_l^*$ may refer to the above sorted list. We use $H_l^*[i,j]$ to denote the (contiguous) sublist $h_i^*,h_{i+1}^*,\ldots,h_j^*$ with $i\leq j$. 

For each $p\in P$, consider a maximal sublist $H_l^*[i,j]$ of points that are covered by the half-plane $p^*$. According to our duality, all points of $B[i,j]$ are hit by $p$. Hence, if $i\neq 1$ and $j\neq t$, then $H_l^*[i,j]$ {\em corresponds to} an arc  of $A_l(p)$, denoted by $\alpha(H_l^*[i,j])$, in the sense that the set of points of $B$ covered by $\alpha(H_l^*[i,j])$ is exactly $B[i,j]$. If $i=1$, then $H_l^*[i,j]$ is {\em covered} by the left arc $\alpha_1(p)$ of $p$ in the sense that all points of $B[1,j]$ are covered by $\alpha_l(p)$. Similarly, if $j=t$, then $H_l^*[i,j]$ is {\em covered} by the right arc $\alpha_2(p)$ of $p$. 

Define $S(p)$ as the set of all maximal sublists of $H_l^*$ covered by $p^*$. Let $S_l(p)$ denote the set of sublists of $S(p)$ excluding the one (if it exists) that contains $h_1^*$ and the one (if it exists) that contains $h_t^*$. According to the above discussion, sublists of $S_l(p)$ correspond exactly to the arcs of $A_l(p)$. 

Following the method in \cite{ref:WangAl24}, we next define a special sublist from $S(p)$, denoted by $H^*_l[i_p,j_p]$. If $H^*_l[i_p,j_p]$ contains neither $h_1^*$ nor $h_t^*$ (i.e., $i_p\neq 1$ and $j_p\neq t$), then we define the lower arc $\alpha_l(p)$ of $p$ as the arc of $A_l(p)$ corresponding to $H^*_l[i_p,j_p]$, i.e., $\alpha_l(p)=\alpha(H^*_l[i_p,j_p])$; otherwise, $\alpha_l(p)$ is undefined. 

\paragraph{Defining $\boldsymbol{H^*_l[i_p,j_p]}$.}
Let $\calU$ denote the upper envelope of the bounding lines of all half-planes of $P^*$. 
For each $p\in P$, depending on whether the bounding line of $\ell_p$ of the half-plane $p^*$ contains an edge of $\calU$, there are two cases to define $H^*_l[i_p,j_p]$. 
\begin{enumerate}
    \item If $\ell_p$ contains an edge $e$ of $\calU$, let $H_l^*[i,j]$ be the sublist of points of $H_l^*$ that are vertically below $e$. Then, all points of $H_l^*[i,j]$ are covered by $p^*$. Hence, $S(p)$ must contain exactly one sublist that contains all points of $H_l^*[i,j]$ and we define $H^*_l[i_p,j_p]$ to be that sublist (see Fig.~\ref{fig:interval_case1}).
    
    \item If $\ell_p$ does not contain any edge of $\calU$, then let $v_p$ be a vertex of $\calU$ that has a tangent line parallel to $\ell_p$. If $S(p)$ contains a sublist that has a point to the left of $v_p$ and also has a point to the right of $v_p$, then we define $H^*_l[i_p,j_p]$ to be that sublist (see Fig.~\ref{fig:interval_case2}); otherwise $H^*_l[i_p,j_p]$ is undefined.  
\end{enumerate}

The above defines $H^*_l[i_p,j_p]$ if it exists. Again, if $i_p\neq 1$ and $j_p\neq t$, then we define $\alpha_l(p)=\alpha(H^*_l[i_p,j_p])$.

\begin{figure}[t]
\begin{minipage}[t]{0.49\linewidth}
\begin{center}
\includegraphics[totalheight=1.1in]{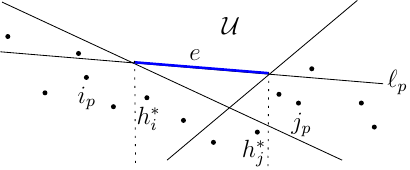}
\caption{Illustration of the definition of $H^*_l[i_p,j_p]$ when $\ell_p$ contains an edge $e$ of $\calU$. $i_p$ and $j_p$ in the figure represent the points $h^*_{i_p}$ and $h^*_{j_p}$, respectively.}
\label{fig:interval_case1}
\end{center}
\end{minipage}
\hspace{0.05in}
\begin{minipage}[t]{0.49\linewidth}
\begin{center}
\includegraphics[totalheight=1.1in]{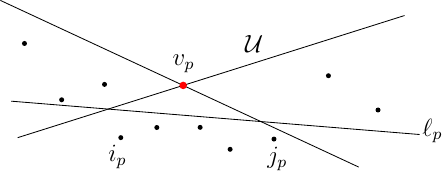}
\caption{Illustration of the definition of $H^*_l[i_p,j_p]$ when $\ell_p$ does not contain an edge of $\calU$. $i_p$ and $j_p$ in the figure represent the points $h^*_{i_p}$ and $h^*_{j_p}$, respectively.}
\label{fig:interval_case2}
\end{center}
\end{minipage}
\end{figure}


\subsection{Correctness}

The following lemma implies that using $\hatcalA$ is sufficient to find a smallest subset of $\calA$ to cover $B$. 

\begin{lemma}\label{lemm:arc_containment}
For any arc $\alpha \in \calA \setminus \hat{\calA}$, $\hat{\calA}$ has an arc containing $\alpha$. 
\end{lemma}
\begin{proof}
Consider an arc $\alpha \in \calA \setminus \hat{\calA}$. Suppose that $\alpha$ is defined by a point $p\in P$, i.e., $\alpha\in A(p)$. Recall that $\{\alpha_1(p),\alpha_2(p)\}$, $A_l(p)$, and $A_u(p)$ form a partition of $A(p)$. As $\alpha\not\in \hat\calA$, which contains both $\alpha_1(p)$ and $\alpha_2(p)$, $\alpha$ is not in $\{\alpha_1(p),\alpha_2(p)\}$. Hence, $\alpha$ is either in $A_l(p)$ or in $A_u(p)$. Without loss of generality, we assume that $\alpha\in A_l(p)$. 

According to our discussion, the arc $\alpha$ corresponds to a sublist $H_l^*[i,j]$ of $S_l(p)$, i.e., $\alpha=\alpha(H_l^*[i,j])$. It has been proved in \cite{ref:WangAl24} that $H_l^*[i,j]$ must be contained in $H_l^*[i_{p'},j_{p'}]$ for some point $p'\in P$. 
Depending on whether $i_{p'}= 1$ and $j_{p'}= t$, there are three cases. 

\begin{itemize}
    \item If $i_{p'}\neq 1$ and $j_{p'}\neq t$, then according to our definition the lower arc $\alpha_l(p')$ is $\alpha(H_l^*[i_{p'},j_{p'}])$, i.e., the arc of $A_l(p')$ corresponding to $H_l^*[i_{p'},j_{p'}]$. Since $H_l^*[i,j]\subseteq H_l^*[i_{p'},j_{p'}]$, we also have $\alpha(H_l^*[i,j])\subseteq \alpha(H_l^*[i_{p'},j_{p'}])$. 
    Since $\alpha=\alpha(H_l^*[i,j])$ and $\alpha(H_l^*[i_{p'},j_{p'}])\in \hatcalA$, we obtain that $\hatcalA$ has an arc containing $\alpha$. 
    
    \item If $i_{p'}= 1$, then $H_l^*[i_{p'},j_{p'}]$ is covered by the left arc $\alpha_1(p')$ of $p'$. Since $H_l^*[i,j]\subseteq H_l^*[i_{p'},j_{p'}]$, $H_l^*[i,j]$ is also covered by $\alpha_1(p')$. Therefore, $\alpha=\alpha(H_l^*[i,j])\subseteq \alpha_1(p')$. As $\alpha_1(p')\in \hatcalA$, we obtain that $\hatcalA$ has an arc containing $\alpha$. 
    
    \item If $j_{p'}= t$, we can follow a similar argument to the above second case. 
\end{itemize}
The lemma is thus proved. 
\end{proof}

In light of Lemma~\ref{lemm:arc_containment}, to solve the circular-point coverage problem on $B$ and $\calA$, we can first compute the arcs of $\hatcalA$ and then find a smallest subset of $\hatcalA$ to cover $B$. The second step takes $O(n\log n)$ time as $|\hatcalA|\leq 4n$. In what follows, we present an algorithm that can compute $\hatcalA$ in $O(n\log n)$ time. 

\subsection{Computing the arcs of $\hatcalA$}
\label{sec:impl}

Recall that $\hatcalA=\bigcup_{p\in P}\hatA(p)$. 
For each $p\in P$, $\hatA(p)$ consists of at most four arcs: a left arc, a right arc, an upper arc, and a lower arc. We first compute all lower and upper arcs.

\paragraph{Computing lower and upper arcs.}
We only discuss how to compute lower arcs since upper arcs can be computed analogously. We first compute $H_l^*$ and $P^*$, which takes $O(n)$ time. Then, we compute the sublists $H^*_l[i_p,j_p]$ for all $p\in P$. For this, Wang and Xue~\cite{ref:WangAl24} gave an $O(n\log n)$ time algorithm by reducing the problem to ray-shooting queries in a simple polygon. With sublists $H^*_l[i_p,j_p]$, all lower arcs can be obtained in $O(n)$ time by following the definition. As such, the lower arcs of all points $p\in P$ can be computed in $O(n\log n)$ time. 

\paragraph{Computing left and right arcs.}
We only discuss how to compute the left arcs since the right arcs can be computed analogously. For each $p\in P$, recall that its left arc $\alpha_l(p)$ is the arc that contains at least one of the two points $b_1$ and $b_n$. As such, if $p$ does not hit $h_1$ or $h_n$, then $\alpha_l(p)$ does not exist. We assume that $p$ hits at least one of $h_1$ and $h_n$. To compute $\alpha_l(p)$, it suffices to find the smallest index $i'_p\in [1,n]$ such that $h_{i'_p}$ is not hit by $p$ and the largest index $j'_p\in [1,n]$ such that $h_{j'_p}$ is not hit by $p$. The following lemma gives an $O(n\log n)$ time algorithm to compute the indices $i'_p$ and $j'_p$ for all $p\in P$. 

\begin{lemma}
\label{lemm:smallestIndex}
There is an algorithm that can compute $i'_p$ and $j'_p$ for all $p \in P$ in $O(n\log n)$ time.
\end{lemma}
\begin{proof}
We only discuss how to compute $i'_p$ as computing $j'_p$ can be done analogously. 

Let $T$ be a complete binary search tree whose leaves from left to right store half-planes of $H_l=\{h_1,h_2,\ldots,h_t\}$ in their index order. As $t\leq n$, the height of $T$ is $O(\log n)$. For each node $v \in T$, let $H_l(v)$ denote the set of half-planes in the leaves of the subtree rooted at $v$. We use $\calL(v)$ to denote the lower envelope of the bounding lines of half-planes of $H_l(v)$. Observe that $p$ hits every half-plane of $H_l(v)$ if and only if $p$ is below $\calL(v)$. Suppose $u$ and $w$ are the left and right children of $v$, respectively. By the definition of the indices of the half-planes of $H_l$, the slopes of the half-planes of $H_l(u)$ are smaller than those of $H_l(w)$. Therefore, $\calL(u)$ and $\calL(w)$ have at most one intersection (indeed, in the dual plane the intersection is dual to the upper tangent of the upper hulls of two sets of points separated by a vertical line). As such, if $\calL(u)$ and $\calL(w)$ are known, then $\calL(v)$ can be computed in $O(|H_l(v)|)$ time. Therefore, the tree $T$ can be constructed in $O(n\log n)$ time in a bottom-up manner.


Consider a point $p \in P$. We compute $i'_p$ using $T$, as follows. We assume that $i'_p\leq t$. If this is not the case, then $i'_p$ can be computed by a similar algorithm using the upper half-planes of $H_u$. 

Starting from the root of $T$, for each node $v$, we do the following. Let $u$ and $w$ be the left and right children of $v$, respectively. We first determine whether $p$ is below $\calL(u)$, which can be done in $O(\log n)$ time by binary search on the sorted list of the vertices of $\calL(u)$ by their $x$-coordinates. If $p$ is above $\calL_u$, then $p$ must be outside a half-plane of $H_l(u)$; in this case, we proceed with $v=u$. Otherwise, we proceed with $v=w$. In this way, $i'_p$ can be computed after searching a root-to-leaf path of $T$. Because we spend $O(\log n)$ time on each node and the height of $T$ is $O(\log n)$, the total time for computing $i'_p$ is $O(\log^2 n)$. This can be improved to $O(\log n)$ using fractional cascading~\cite{ref:ChazelleFr86}, as follows.

The $x$-coordinates of all vertices of $\calL(v)$ partition the $x$-axis into a set $\calI_v$ of intervals. To determine whether $p$ is above $\calL_v$, it suffices to find the interval of $\calI_v$ that contains $x(p)$, the $x$-coordinate of $p$. 
We construct a fractional cascading data structure on the intervals of $\calI_v$ of all nodes $v \in T$, which takes $O(n \log n)$ time~\cite{ref:ChazelleFr86} since the total number of such intervals is $O(n \log n)$.
With the fractional cascading data structure, we only need to do binary search on the set of the intervals stored at the root to determine the interval containing $x(p)$, which takes $O(\log n)$ time. After that, for each node $v$ during the algorithm described above, the interval of $\calI_v$ containing $x(p)$ can be determined in $O(1)$ time~\cite{ref:ChazelleFr86}. As such, computing $i'_p$ takes $O(\log n)$ time. 

In summary, the indices $i'_p$ for all $p\in P$ can be computed in $O(n\log n)$ time. 
\end{proof}

The above computes all arcs of $\hatcalA$ in $O(n\log n)$ time. 
The following theorem summarizes the main result of this paper.  

\begin{theorem}
\label{theo:hitting}
Given a set $P$ of $n$ points and a set $H$ of $n$ half-planes, 
a smallest hitting set of $H$ can be computed in $O(n \log n)$ time.
\end{theorem}

It would be interesting to see whether our techniques can be used to tackle the half-plane coverage problem. The currently best algorithm for the problem runs in $O(n^{4/3} \log^{5/3} n \log^{O(1)} \log n)$ time~\cite{ref:WangAl24}.

\bibliographystyle{plainurl}

\end{document}